\documentclass[submission,copyright,creativecommons]{eptcs}
\providecommand{\event}{GASCOM 2026} 

\newcommand{\wtc}[2]{\widetilde{C}_{#1}}
\newcommand{\hnq}{H_n(q)}
\newcommand{\sn}{\mathfrak{S}_n}
\newcommand{\zsn}{\mathbb{Z}[\sn]}
\newcommand{\dotsc}{\ldots}
\newcommand{\mathbb}{\mathbf}
\newcommand{\mfs}[1]{\mathfrak{S}_{#1}}
\newcommand{\defeq}{:=}

\newcommand{\avoidsp}{avoids the patterns $3412$ and $4231${}}
\newcommand{\avoidp}{avoid the patterns $3412$ and $4231${}}
\newcommand{\avoidingp}{avoiding the patterns $3412$ and $4231${}}

\newcommand{\ntnsp}{\negthinspace}
\newcommand{\ntksp}{\negthickspace}
\newcommand{\nTksp}{\negthickspace\negthickspace}

\newcommand{\ssm}{\setminus}

\newcommand{\net}[1]{\mathcal F_{#1}}
\newcommand{\singdeg}{\mathrm{singdeg}}
\newcommand{\pgap}{\mathrm{gap_{3412}}}
\newcommand{\cnet}[1]{\mathcal F^\bullet_{#1}}

\newcommand{\phm}{\phantom M}

\newcommand{\type}{\mathrm{type}}

\newcommand{\zqq}{\mathbb{Z}[\qp12, \qm12]}
\newcommand{\qp}[2]{q^{\frac{#1}{#2}}}
\newcommand{\qm}[2]{q^{\negthinspace\Bar\,\frac{#1}{#2}}}
\newcommand{\ol}[1]{\overline{#1}}
\newcommand{\dfct}{\mathrm{dfct}}
\newcommand{\ctr}[1]{{x}_{#1}}
\newcommand{\vertex}[1]{{#1}}

\usepackage{eufrak}
\usepackage{iftex}
\usepackage{mathrsfs}
\usepackage{tikz}
\usepackage{amsmath}
\usepackage{amsthm}
\usetikzlibrary{arrows}

\newtheorem{thm}{Theorem}[section]
\newtheorem{prop}[thm]{Proposition}

\newtheorem{conj}[thm]{Conjecture}

\newtheorem{defn}[thm]{Definition}
\newtheorem{alg}[thm]{Algorithm}

\ifpdf
  \usepackage{underscore}         
  \usepackage[T1]{fontenc}        
\else
  \usepackage{breakurl}           
\fi

\title{On the Impossibility of Parabolic Factorization of certain Kazhdan-Lusztig Basis Elements}
\author{Tommy Parisi
\qquad 
Mark Skandera \qquad Ben Spahiu
\qquad Jiayuan Wang
\institute{Lehigh University, Bethlehem PA, USA}
\email{tjp225@lehigh.edu \quad mas906@lehigh.edu \quad bes226@lehigh.edu \quad jiw922@lehigh.edu}
}
\def\titlerunning{Impossibility of Parabolic Factorization of certain Kazhdan-Lusztig Basis Elements}
\def\authorrunning{T. Parisi, M. Skandera, B. Spahiu, \& J. Wang}
\begin{document}
\maketitle

\begin{abstract}
For $w$ in the symmetric group $\sn$, let $\wtc wq$ be the corresponding modified, signless Kazhdan-Lusztig basis element of the type-$A$ Hecke algebra $\hnq$. An extension 
[{\em Ann.\;Comb.}\;{\bf 25}, no.\,3 (2021) pp.~757--787] of a result 
of Deodhar [{\em Geom.\;Dedicata} {\bf 36}, (1990) pp.~95--119] implies that any factorization of the form
    $\wtc wq = \frac1{f(q)} \wtc{v^{(1)}}q \cdots \wtc{v^{(r)}}q$,
with $v^{(1)},\dotsc,v^{(r)}$ maximal elements of parabolic subgroups of $\sn$ and $f(q) \in \mathbb N[q]$ depending on these, provides cancellation-free combinatorial
interpretations of the polynomials $\{P_{v,w}(q) \,|\, v \in \sn \}$ appearing
in the expansion $\sum_v P_{v,w}(q) T_v$ 
of $\wtc wq$ in terms of the natural basis $\{ T_v \,|\, v \in \sn \}$ of $\hnq$. While the set of permutations $w \in \sn$ admitting such a factorization of $\smash{\wtc wq}$ has not yet been characterized,
we apply a result of Gaetz--Gao 
[{\em Adv. Math.}
{\bf 457} (2024)
Paper No. 109941] to describe a set for which such a factorization cannot exist.
\end{abstract}

\section{Introduction}\label{sec: introduction}
The {\em Kazhdan--Lusztig polynomials} $\{ P_{v,w}(q) \,|\, v, w \in \sn \} \subset \mathbb N[q]$ are entries of the change-of-basis matrix relating a certain  {\em Kazhdan--Lusztig basis} of the Hecke algebra with another {\em natural basis.}
First appearing in the study of representations of the Hecke algebra, 
they were given existential and recursive definitions in \cite{KLRepCH}.
Appearances of the polynomials in other areas such as Lie Theory~\cite{BeilBern}, \cite{BeilBernPfofJantzen}, \cite{BryKash},
quantum groups~\cite{FKK}, 
and Schubert varieties~\cite{KLRepCH}, \cite{KLSchub} 
have inspired a search for simpler descriptions.
Ideally, such a description should interpret each coefficient of $P_{v,w}(q)$ as a set cardinality.

Some famous alternative formulas for the Kazhdan--Lusztig polynomials are due to Brenti and Deodhar. Brenti expressed $P_{v,w}(q)$ in two different ways
as simple linear combinations of recursively defined polynomials in $\mathbb Z[q]$ having both positive and negative 
coefficients~\cite[\S 3]{Brenti94}, \cite[\S 3]{Brenti97}.
Because of negative coefficients and recursive definitions, these formulas do not interpret coefficients in $P_{v,w}(q)$ as set cardinalities. Deodhar~\cite{Deodhar90} developed an algorithm which takes any reduced expression for $w$ as an input, and outputs a set $\mathscr E_{\min}$ of (not necessarily reduced) expressions
for other permutations in $\sn$. For each $v \in \sn$ and $k > 0$, the coefficient of $q^k$ in $P_{v,w}(q)$ is equal to the cardinality of a certain subset of $\mathscr E_{\min}$. On the other hand, the algorithmic component of Deodhar's method makes it difficult to apply his combinatorial interpretation in practice.


Billey and Warrington showed~\cite[Thm.\,1, Rmk.\,6]{BWHex} that 
when $w$ has certain properties, Deodhar's algorithm is trivial, 
and the output set 
$\mathscr E_{\min}$ of expressions can be replaced by a more visually appealing set of path families in a certain wiring diagram.  
Again for each $v$ and $k$, the coefficient of $q^k$ in $P_{v,w}(q)$ is equal to the cardinality of a subset of these path families. Clearwater and the second author~\cite[Cor.\,5.3]{CSkanTNNChar} then extended this result to permutations $w$ for which the Kazhdan--Lusztig basis element $\smash{\wtc wq}$ factors nicely, but did not solve the problem~\cite[Quest.\,4.5]{SkanNNDCB} of characterizing such permutations $w$. 

In Sections~\ref{s:snplanar} -- \ref{s:hnqplanar} we review basic facts about
the symmetric group, planar networks, the Hecke algebra, and the Kazhdan--Lusztig basis and polynomials.
In Section~\ref{s:defectremoval} we use the result~\cite[Cor.\,5.3]{CSkanTNNChar} to state properties of polynomials which arise in the natural expansion of products of certain Kazhdan--Lusztig basis elements of the Hecke algebra. This leads to a partial answer in Section~\ref{s:nofactor} to the characterization question~\cite[Quest.\,4.5]{SkanNNDCB}: a description of 
certain Kazhdan--Lusztig basis elements which do not factor as desired. 

\section{The symmetric group and planar networks}\label{s:snplanar}

Let $\sn$ be the symmetric group, with standard generators
$s_1,\dotsc,s_{n-1}$, length function $\ell$, and Bruhat order $\leq$.
(See, e.g., \cite{BBCoxeter} for definitions.)
Given a word $u = u_1 \cdots u_k$ in $\mfs k$,
and a word $y = y_1 \cdots y_k$ having $k$ distinct letters,
we say that {\em $y$ matches the pattern $u$} if the letters of $y$ appear
in the same relative order as those of $u$; that is, if we have
$u_i < u_j$ if and only if $y_i < y_j$ for all $i,j \in [k] \defeq \{1,\dotsc,k\}$.
On the other hand, 
say that $w \in \mfs n$ 
{\em avoids the pattern $u$} if no subword $y$
of $w$ matches the pattern $u$.
For example, the subword $423$ of the permutation $4231$ matches the pattern $312$; thus $4231$ does not avoid the pattern $312$.

It is easy to see that for 
each subinterval $[a,b] \defeq \{a,\dotsc,b\}$ of $[n]$,
the 
{\em reversal} 
\begin{equation}\label{eq:reversaldef} s_{[a,b]} \defeq 1 \cdots (a-1) b \cdots a (b+1) \cdots n \in \sn
\end{equation}
\avoidsp.
This element
is the unique longest (maximum length) element of the subgroup of $\sn$ generated by $s_a, \dotsc, s_{b-1}$.
More generally, each {\em parabolic} subgroup of $\sn$ generated by a subset 
of generators 
has longest element equal to a product of reversals on disjoint intervals.
Multiplication of reversals in $\sn$ or of related elements
\begin{equation}\label{eq:Dab}
D_{[a,b]} \defeq \ntksp 
\sum_{v \leq s_{[a,b]}}\ntksp  v
\end{equation}
in $\zsn$ can be performed graphically with certain planar networks.


Define a {\em planar network of order n} to be a directed, planar, acyclic multigraph with $2n$ boundary vertices having $n$ source vertices on the left and $n$ sink vertices on the right, both labeled $1, \dotsc, n$ from bottom to top.
We will allow edges $(\vertex x, \vertex y)$
to be marked by a positive integer multiplicity $m(x,y)$.
Let $\net n$ denote the set of such networks.
For each subinterval $[a,b]$ of $[n]$ we define a {\em simple star network}
$\smash{F_{[a,b]}} \in \net n$ by
\begin{enumerate}
    \item an interior vertex $x$ lies between the sources and sinks,
    \item for $i \in [a,b]$ we have edges from source $i$ to $x$ and from $x$ to sink $i$,
    \item for $i \notin [a,b]$ we have edges from source $i$ to sink $i$.
\end{enumerate}
For example, the simple star network $F_{\smash{[2,4]}} \in \net 4$ is
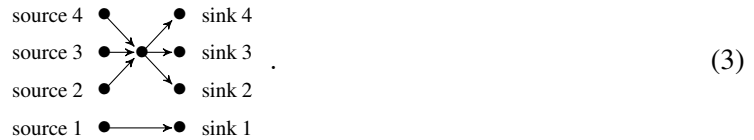
\begin{equation}
\begin{tikzpicture}[scale=.5,baseline=15]
\node at (-2.5,2.5) {$\scriptstyle{ \mathrm{source}\;4}$};
\node at (-2.5,1.5) {$\scriptstyle{ \mathrm{source}\;3}$};
\node at (-2.5,0.5) {$\scriptstyle{ \mathrm{source}\;2}$};
\node at (-2.5,-0.5) {$\scriptstyle{ \mathrm{source}\;1}$};  
\node at (2.25,2.5) {$\scriptstyle{ \mathrm{sink}\;4}$};
\node at (2.25,1.5) {$\scriptstyle{ \mathrm{sink}\;3}$};
\node at (2.25,0.5) {$\scriptstyle{ \mathrm{sink}\;2}$};
\node at (2.25,-0.5) {$\scriptstyle{ \mathrm{sink}\;1}$};  
\draw[->,>=stealth'] (-1,2.5) -- (-.2,1.65);
\draw[->,>=stealth'] (-1,1.5) -- (-.2,1.5);
\draw[->,>=stealth'] (-1,0.5) -- (-.2,1.35);
\draw[->,>=stealth'] (0,1.5) -- (.8,2.35);
\draw[->,>=stealth'] (0,1.5) -- (.8,1.5);
\draw[->,>=stealth'] (0,1.5) -- (.8,0.65);
\draw[->,>=stealth'] (-1,-0.5) -- (.8,-0.5);
\node at (-1,2.5) {$\bullet$}; 
\node at (-1,1.5) {$\bullet$}; 
\node at (-1,0.5) {$\bullet$}; 
\node at (-1,-0.5) {$\bullet$};
\node at (1,2.5) {$\bullet$}; 
\node at (1,1.5) {$\bullet$}; 
\node at (1,0.5) {$\bullet$}; 
\node at (1,-0.5) {$\bullet$};
\node at (0,1.5) {$\bullet$};
\end{tikzpicture}
\ .
\end{equation}  

For economy, we will omit edge orientations
and the words "source" and "sink" from figures.  
Thus
the seven simple star networks in $\net 4$ are
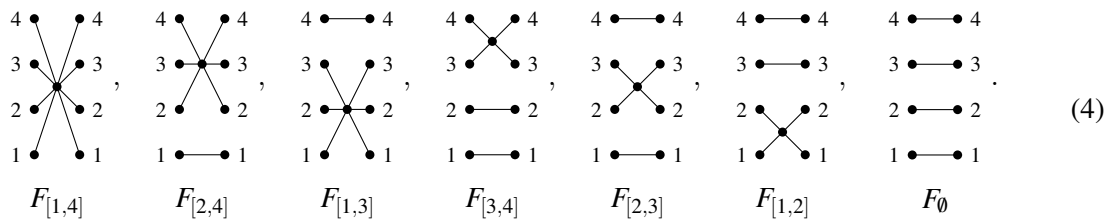
\begin{equation}\label{eq:simplestarnets}
\begin{gathered}
   \begin{tikzpicture}[scale=.6,baseline=-25]
\node at (-.4,0) {$\scriptstyle 4$};
\node at (-.4,-1) {$\scriptstyle 3$};
\node at (-.4,-2) {$\scriptstyle{2}$};
\node at (-.4,-3) {$\scriptstyle{1}$};  
\node at (1.4,0) {$\scriptstyle 4$};
\node at (1.4,-1) {$\scriptstyle 3$};
\node at (1.4,-2) {$\scriptstyle{2}$};
\node at (1.4,-3) {$\scriptstyle{1}$};  
\draw[-] (0,0) -- (1,-3);
\draw[-] (0,-1) -- (1,-2);
\draw[-] (0,-2) -- (1,-1);
\draw[-] (0,-3) -- (1,0);
\fill (0,0) circle  (1mm); \fill (0,-1) circle  (1mm); \fill (0,-2) circle  (1mm); \fill (0,-3) circle  (1mm);
\fill (1,0) circle  (1mm); \fill (1,-1) circle  (1mm); \fill (1,-2) circle  (1mm); \fill (1,-3) circle  (1mm);
\fill (.5, -1.5) circle (1mm);
\end{tikzpicture}
,\\
\phantom{\sum}\ntksp F_{[1,4]}\phantom{\sum}
\end{gathered}
\phm
\begin{gathered}
\begin{tikzpicture}[scale=.6,baseline=-25]
\node at (-.4,0) {$\scriptstyle 4$};
\node at (-.4,-1) {$\scriptstyle 3$};
\node at (-.4,-2) {$\scriptstyle{2}$};
\node at (-.4,-3) {$\scriptstyle{1}$};  
\node at (1.4,0) {$\scriptstyle 4$};
\node at (1.4,-1) {$\scriptstyle 3$};
\node at (1.4,-2) {$\scriptstyle{2}$};
\node at (1.4,-3) {$\scriptstyle{1}$}; 
\draw[-] (0,0) -- (1,-2);
\draw[-] (0,-1) -- (1,-1);
\draw[-] (0,-2) -- (1,0);
\draw[-] (0,-3) -- (1,-3);
\fill (0,0) circle  (1mm); \fill (0,-1) circle  (1mm); \fill (0,-2) circle  (1mm); \fill (0,-3) circle  (1mm);
\fill (1,0) circle  (1mm); \fill (1,-1) circle  (1mm); \fill (1,-2) circle  (1mm); \fill (1,-3) circle  (1mm);
\fill (.5, -1) circle (1mm);
\end{tikzpicture},\\
   \phantom{\sum}\ntksp F_{[2,4]}\phantom{\sum}
 \end{gathered}
\phm
\begin{gathered}
\begin{tikzpicture}[scale=.6,baseline=-25]
\node at (-.4,0) {$\scriptstyle 4$};
\node at (-.4,-1) {$\scriptstyle 3$};
\node at (-.4,-2) {$\scriptstyle{2}$};
\node at (-.4,-3) {$\scriptstyle{1}$};  
\node at (1.4,0) {$\scriptstyle 4$};
\node at (1.4,-1) {$\scriptstyle 3$};
\node at (1.4,-2) {$\scriptstyle{2}$};
\node at (1.4,-3) {$\scriptstyle{1}$}; 
\draw[-] (0,0) -- (1,0);
\draw[-] (0,-1) -- (1,-3);
\draw[-] (0,-2) -- (1,-2);
\draw[-] (0,-3) -- (1,-1);
\fill (0,0) circle  (1mm); \fill (0,-1) circle  (1mm); \fill (0,-2) circle  (1mm); \fill (0,-3) circle  (1mm);
\fill (1,0) circle  (1mm); \fill (1,-1) circle  (1mm); \fill (1,-2) circle  (1mm); \fill (1,-3) circle  (1mm);
\fill (.5, -2) circle (1mm);
\end{tikzpicture},\\
   \phantom{\sum}\ntksp F_{[1,3]}\phantom{\sum}
 \end{gathered}
\phm
\begin{gathered}
\begin{tikzpicture}[scale=.6,baseline=-25]
\node at (-.4,0) {$\scriptstyle 4$};
\node at (-.4,-1) {$\scriptstyle 3$};
\node at (-.4,-2) {$\scriptstyle{2}$};
\node at (-.4,-3) {$\scriptstyle{1}$};  
\node at (1.4,0) {$\scriptstyle 4$};
\node at (1.4,-1) {$\scriptstyle 3$};
\node at (1.4,-2) {$\scriptstyle{2}$};
\node at (1.4,-3) {$\scriptstyle{1}$}; 
\draw[-] (0,0) -- (1,-1);
\draw[-] (0,-1) -- (1,0);
\draw[-] (0,-2) -- (1,-2);
\draw[-] (0,-3) -- (1,-3);
\fill (0,0) circle  (1mm); \fill (0,-1) circle  (1mm); \fill (0,-2) circle  (1mm); \fill (0,-3) circle  (1mm);
\fill (1,0) circle  (1mm); \fill (1,-1) circle  (1mm); \fill (1,-2) circle  (1mm); \fill (1,-3) circle  (1mm);
\fill (.5, -0.5) circle (1mm);
\end{tikzpicture},\\
   \phantom{\sum}\ntksp F_{[3,4]}\phantom{\sum}
 \end{gathered}
\phm
\begin{gathered}
\begin{tikzpicture}[scale=.6,baseline=-25]
\node at (-.4,0) {$\scriptstyle 4$};
\node at (-.4,-1) {$\scriptstyle 3$};
\node at (-.4,-2) {$\scriptstyle{2}$};
\node at (-.4,-3) {$\scriptstyle{1}$};  
\node at (1.4,0) {$\scriptstyle 4$};
\node at (1.4,-1) {$\scriptstyle 3$};
\node at (1.4,-2) {$\scriptstyle{2}$};
\node at (1.4,-3) {$\scriptstyle{1}$}; 
\draw[-] (0,0) -- (1,0);
\draw[-] (0,-1) -- (1,-2);
\draw[-] (0,-2) -- (1,-1);
\draw[-] (0,-3) -- (1,-3);
\fill (0,0) circle  (1mm); \fill (0,-1) circle  (1mm); \fill (0,-2) circle  (1mm); \fill (0,-3) circle  (1mm);
\fill (1,0) circle  (1mm); \fill (1,-1) circle  (1mm); \fill (1,-2) circle  (1mm); \fill (1,-3) circle  (1mm);
\fill (.5, -1.5) circle (1mm);
\end{tikzpicture},\\
   \phantom{\sum}\ntksp F_{[2,3]}\phantom{\sum}
 \end{gathered}
\phm
\begin{gathered}
\begin{tikzpicture}[scale=.6,baseline=-25]
\node at (-.4,0) {$\scriptstyle 4$};
\node at (-.4,-1) {$\scriptstyle 3$};
\node at (-.4,-2) {$\scriptstyle{2}$};
\node at (-.4,-3) {$\scriptstyle{1}$};  
\node at (1.4,0) {$\scriptstyle 4$};
\node at (1.4,-1) {$\scriptstyle 3$};
\node at (1.4,-2) {$\scriptstyle{2}$};
\node at (1.4,-3) {$\scriptstyle{1}$}; 
\draw[-] (0,0) -- (1,0);
\draw[-] (0,-1) -- (1,-1);
\draw[-] (0,-2) -- (1,-3);
\draw[-] (0,-3) -- (1,-2);
\fill (0,0) circle  (1mm); \fill (0,-1) circle  (1mm); \fill (0,-2) circle  (1mm); \fill (0,-3) circle  (1mm);
\fill (1,0) circle  (1mm); \fill (1,-1) circle  (1mm); \fill (1,-2) circle  (1mm); \fill (1,-3) circle  (1mm);
\fill (.5, -2.5) circle (1mm);
\end{tikzpicture},\\
   \phantom{\sum}\ntksp F_{[1,2]}\phantom{\sum}
 \end{gathered}
\
\phm
\begin{gathered}
\begin{tikzpicture}[scale=.6,baseline=-25]
\node at (-.4,0) {$\scriptstyle 4$};
\node at (-.4,-1) {$\scriptstyle 3$};
\node at (-.4,-2) {$\scriptstyle{2}$};
\node at (-.4,-3) {$\scriptstyle{1}$};  
\node at (1.4,0) {$\scriptstyle 4$};
\node at (1.4,-1) {$\scriptstyle 3$};
\node at (1.4,-2) {$\scriptstyle{2}$};
\node at (1.4,-3) {$\scriptstyle{1}$}; 
\draw[-] (0,0) -- (1,0);
\draw[-] (0,-1) -- (1,-1);
\draw[-] (0,-2) -- (1,-2);
\draw[-] (0,-3) -- (1,-3);
\fill (0,0) circle  (1mm); \fill (0,-1) circle  (1mm); \fill (0,-2) circle  (1mm); \fill (0,-3) circle  (1mm);
\fill (1,0) circle  (1mm); \fill (1,-1) circle  (1mm); \fill (1,-2) circle  (1mm); \fill (1,-3) circle  (1mm);
\end{tikzpicture}.\\
\phantom{\sum}\ntksp F_\emptyset \phantom{\sum}
 \end{gathered}
\end{equation}

Given networks $E, F \in \net{n}$,
in which all sources have
outdegree $1$ and all sinks have indegree $1$, define
the concatenation $E \circ F$ of $E$ and $F$
as follows.  For $i = 1,\dotsc, n$, do
\begin{enumerate}
\item remove sink $i$ of $E$ and source $i$ of $F$,
\item merge each edge $(\vertex x, \text{sink }i)$ in $E$ with each edge
  $( \text{source }i, \vertex y )$ in $F$
  to form a single edge $(\vertex x, \vertex y)$ in $E \circ F$.
\end{enumerate}
Thus a concatenation of the form
$F_{[a_1,b_1]} \circ \cdots \circ F_{[a_m,b_m]} \in \net n$
has $2n + m$ edges:
$n$ sources inherited from $F_{[a_1,b_1]}$, $n$ sinks inherited from $F_{[a_m,b_m]}$, and $m$ internal vertices 
$x_1,\dotsc, x_m$, where $x_j$ is inherited 
from $F_{[a_j,b_j]}$.
Sometimes in a concatenation $E \circ F$,
there may exist internal vertices $\vertex x$ in $E$, $\vertex y$ in $F$
with
$m(\vertex x, \vertex y) > 1$ multiplicity-$1$ edges
incident upon both.
Define the {\em condensed concatenation}
$E \bullet F$ to be the
subdigraph of $E \circ F$ obtained
by removing, for all such pairs $(\vertex x, \vertex y)$,
all but one of the $m(\vertex x, \vertex y)$ edges incident upon both,
and by marking this edge with the multiplicity $m(\vertex x, \vertex y)$.
For example, in $\net 4$ we have the 
graphs
\begin{equation}\label{eq:allbutone2}
  F_{[1,3]} \circ F_{[2,4]} \circ F_{[1,3]} =
\begin{tikzpicture}[scale=.6,baseline=0]
\node at (-.4,1.5) {$\scriptstyle 4$};
\node at (-.4,0.5) {$\scriptstyle 3$};
\node at (-.4,-0.5) {$\scriptstyle{2}$};
\node at (-.4,-1.5) {$\scriptstyle{1}$};  
\node at (3.4,1.5) {$\scriptstyle 4$};
\node at (3.4,0.5) {$\scriptstyle 3$};
\node at (3.4,-0.5) {$\scriptstyle{2}$};
\node at (3.4,-1.5) {$\scriptstyle{1}$};  
\draw[-] (0,1.5) -- (1,1.5) -- (2,-0.5) -- (3,-0.5);
\draw[-] (0,0.5) -- (1,-1.5) -- (2,-1.5) -- (3,0.5);
\draw[-] (0,-0.5) -- (1,-0.5) -- (2,1.5) -- (3,1.5);
\draw[-] (0,-1.5) -- (1,0.5) -- (2,0.5) -- (3,-1.5);
\fill (0,1.5) circle  (1mm); \fill (0,.5) circle  (1mm); \fill (0,-.5) circle  (1mm); \fill (0,-1.5) circle  (1mm);
\fill (3,1.5) circle  (1mm); \fill (3,.5) circle  (1mm); \fill (3,-.5) circle  (1mm); \fill (3,-1.5) circle  (1mm);
\fill (.5, -.5) circle (1mm); \fill (1.5, .5) circle (1mm); \fill (2.5, -.5) circle (1mm);
\end{tikzpicture},
\qquad
  F_{[1,3]} \bullet F_{[2,4]} \bullet F_{[1,3]} =
\begin{tikzpicture}[scale=.6,baseline=0]
\node at (-0.9,1.5) {$\scriptstyle 4$};
\node at (-0.9,0.5) {$\scriptstyle 3$};
\node at (-0.9,-0.5) {$\scriptstyle{2}$};
\node at (-0.9,-1.5) {$\scriptstyle{1}$};  
\node at (2.05,-.25) {$\scriptstyle _{(2)}$};
\node at (.95,-.25) {$\scriptstyle _{(2)}$};
\node at (3.9,1.5) {$\scriptstyle 4$};
\node at (3.9,0.5) {$\scriptstyle 3$};
\node at (3.9,-0.5) {$\scriptstyle{2}$};
\node at (3.9,-1.5) {$\scriptstyle{1}$};  
\draw[-] (-0.5,1.5) -- (.7,1.5) -- (1.5,0.5) -- (2.3,1.5) -- (3.5,1.5);
\draw[-] (-0.5,0.5) -- (-0,-0.5) -- (1.5,0.5) -- (3,-0.5) -- (3.5,0.5);
\draw[-] (-0.5,-0.5) -- (-0,-0.5) -- (1.5,0.5) -- (3,-0.5) -- (3.5,-0.5);
\draw[-] (-0.5,-1.5) -- (-0,-0.5) -- (.5,-1.5) -- (2.5,-1.5) -- (3,-0.5) -- (3.5,-1.5);
\fill (-.5,1.5) circle  (1mm); \fill (-.5,.5) circle  (1mm); \fill (-.5,-.5) circle  (1mm); \fill (-.5,-1.5) circle  (1mm);
\fill (3.5,1.5) circle  (1mm); \fill (3.5,.5) circle  (1mm); \fill (3.5,-.5) circle  (1mm); \fill (3.5,-1.5) circle  (1mm);
\fill (0, -.5) circle (1mm); 
\fill (1.5, .5) circle (1mm); \fill (3, -.5) circle (1mm);
\node at (0.15, 0.1) {$\scriptstyle x_1$};
\node at (1.5, 1.1) {$\scriptstyle x_2$};
\node at (2.8, 0.1) {$\scriptstyle x_3$};
\end{tikzpicture},
\end{equation}
in which the two multiplicity-$2$ edges $(x_1,x_2)$,
$(x_2,x_3)$ of 
$F_{[1,3]}\bullet F_{[2,4]} \bullet F_{[1,3]}$
are the remnants of pairs of edges incident upon the same 
internal vertices in $F_{[1,3]} \circ F_{[2,4]} \circ F_{[1,3]}$.

Define a {\em star network} to be an element of $\net n$ constructed by
concatenation or condensed concatenation of simple star networks.
Let $\cnet n$ denote the subset of $\net n$ consisting of condensed concatenations
of finitely many simple star networks. 
Call a sequence $\pi = (\pi_1,\dotsc,\pi_n)$
of source-to-sink paths in a star network $F \in \cnet n$
a {\em path family of type $v = v_1 \cdots v_n \in \sn$}
 if for all $i$, path $\pi_i$ begins at source $i$ and terminates at sink $v_i$.  Say that $\pi$ {\em covers} $F$ if each edge $(x_i,x_j)$ of $F$ belongs to $m(x_i,x_j)$ of the paths in $\pi$, and define the sets
\begin{equation}\label{eq:pathfams}
  \begin{gathered}
    \Pi(F) = \{ \pi \,|\, \pi \text{ a path family covering } F \},\\
    \Pi_v(F) = \{ \pi \in \Pi(F) \,|\, \type(\pi) = v \}.
    \end{gathered}
\end{equation}

In terms of the definitions \eqref{eq:pathfams}, we may combinatorially interpret products of 
elements (\ref{eq:Dab}) quite simply.  
We say that $F$ {\em graphically represents}
\begin{equation}\label{eq:PiwF}
    \sum_{v \in \sn} 
    \ntksp |\Pi_v(F)|\,v
\end{equation}
{\em as an element of $\zsn$}.
For
$F = F_{[a_1,b_1]} \circ \cdots \circ F_{[a_m,b_m]}$,
this element is
$D_{[a_1,b_1]} \cdots D_{[a_m,b_m]}$;
for
$F = F_{[a_1,b_1]} \bullet \cdots \bullet F_{[a_m,b_m]}$,
it is
$D_{[a_1,b_1]} \cdots D_{[a_m,b_m]}$
divided by the product, over all edges $(x_i,x_j)$,  of the numbers $m(x_i,x_j)!$.

\section{The Hecke algebra and planar networks}\label{s:hnqplanar}

Define the {\em (type-$A$ Iwahori--) Hecke algebra} $\hnq$
to be the
$\zqq$-span of its {\em natural basis} $\{ T_w \,|\, w \in \sn \}$,
with multiplication given by $T_e T_w = T_w$ where $e = 1 \cdots n$ is the identity element of $\sn$, and
\begin{equation*}
    T_{s_i} T_w = \begin{cases}
        T_{s_i w} &\text{if $s_i w > w$},\\
        (q-1)T_{s_i w} + q T_w &\text{if $s_i w < w$}.
        \end{cases}
\end{equation*}
Specializing at $\qp12 = 1$ we have $H_n(1) \cong \zsn$
with $T_w \mapsto w$.

A semilinear involution on $\hnq$, known as the {\em bar involution}, is defined by
\begin{equation*}
      \ol{\qp12} = \qm12, \qquad \ol{T_w} = (T_{w^{-1}})^{-1}, \qquad
\overline{\sum_{w \in \sn} B_w(q) T_w}
  = \sum_{w \in \sn} \ol{B_w(q)}\, \ol{T_w}.
\end{equation*}
Kazhdan and Lusztig showed~\cite{KLRepCH} that $\hnq$ has a unique
basis $\{ C'_w \,|\, w \in \sn \}$
satisfying 
$\ol{C'_w} = C'_w$ 
and
\begin{equation}\label{eq:klbasis}
\qp{\ell(w)}2 C'_w = \sum_{v \leq w} P_{v,w}(q) T_w,
\end{equation}
where coefficients $P_{v,w}(q) \in \mathbb Z[q]$
satisfy
$\deg(P_{v,w}(q)) < \frac{\ell(w) - \ell(v)-1}2$ for $v < w$, 
and
$P_{w,w}(q) = 1$ for all $w$.
It is known that these 
{\em Kazhdan--Lusztig polynomials} 
satisfy $P_{v,w}(q) \in \mathbb N[q]$, and $P_{v,w}(0) = 1$ for $v \leq w$.  
We also have~\cite{LakSan} that if
$w$ \avoidsp,
then $P_{v,w}(q) = 1$ for all $v \leq w$.
For convenience, we define 
\begin{equation}\label{eq: change of basis}
\wtc wq \defeq \qp{\ell(w)}2 C'_w.
\end{equation}
Kazhdan--Lusztig basis elements and their products
appear in various settings, including
intersection homology~\cite{BBDFaisceaux}, \cite{SpringerQACI}, algorithmic and combinatorial description of Kazhdan--Lusztig basis elements themselves~\cite{BWHex}, \cite{Deodhar90}, Schubert varieties~\cite{BWHex}, total nonnegativity~\cite{GJImm}, \cite{SkanNNDCB}, \cite{StemImm}, \cite{StemConj}, 
trace evaluations~\cite{CHSSkanEKL}, \cite{CSkanTNNChar}, \cite{GreeneImm}, \cite{KLSBasesQMBIndSgn}, \cite{SkanHyperGC}, and chromatic symmetric functions~\cite{CHSSkanEKL}, \cite{SkanHyperGC}.

Deodhar~\cite[Prop.\,3.5]{Deodhar90}
studied sequences
$(s_{i_1},\dotsc,s_{i_k})$
of generators of $\sn$, 
products of the 
corresponding 
Kazhdan--Lusztig basis elements 
$\wtc{\smash{s_{i_j}}}q = T_e + T_{\smash{s_{i_j}}}$ of $\hnq$, 
and their natural expansions
\begin{equation}\label{eq:Deoprod}
\wtc{s_{i_1}\ntksp}q \cdots \wtc{s_{i_m}\ntksp}q
= \sum_{v \in \sn} A_v(q) T_v.
\end{equation}
He 
described the 
coefficients $\{ A_v(q) \,|\, v \in \sn \} \subset \mathbb Z[q]$ 
in terms of {\em subexpressions}
of $(s_{i_1}, \dotsc, s_{i_m})$,
sequences $\sigma = (
\sigma_1, \dotsc, \sigma_m)$ 
with $\sigma_j \in \{ e, s_{i_j}\}$
for $j = 1,\dotsc, m$.
(Our treatment here differs slightly from that of \cite{Deodhar90} but is equivalent.) Call index $j$ a {\em defect} of $\sigma$ if 
\begin{equation}\label{eq:deodhardefect}
\sigma_1 \cdots \sigma_{j-1}s_{i_j} < \sigma_1 \cdots \sigma_{j-1}
\end{equation}
and let $\dfct(\sigma)$ denote the number of defects of $\sigma$. 
(Observe that $j=1$ cannot be a defect: we have $s_{i_1} > e$ always.) 
Each coefficient on the right-hand side of (\ref{eq:Deoprod}) is given by
\begin{equation}\label{eq:Deodefectformula}
    A_v(q) = \sum_\sigma q^{\dfct(\sigma)},
\end{equation}
where the sum is over all subexpressions
$\sigma$ of 
$(s_{i_1},\dotsc, s_{i_m})$
satisfying $\sigma_1 \cdots \sigma_m = v$.


Billey and Warrington observed~\cite[Rmk.\,6]{BWHex}
that 
the defect statistic has a simple graphical interpretation.
Specifically, subexpressions 
of $(s_{i_1},\dotsc,s_{i_m})$ correspond bijectively to path families
covering 
\begin{equation}\label{eq:wirediagprod}
F =
F_{[i_1,i_1+1]} \bullet \cdots \bullet F_{[i_m,i_m+1]} 
\end{equation}
in $\cnet n$ with 
$(\sigma_1,\dotsc,\sigma_m)$ corresponding 
to the family $\pi \in \Pi(F)$
constructed by prescribing 
\begin{equation*}
\text{the paths meeting at $\ctr j$ }
\begin{cases}
    \text{cross there} &\text{if $\sigma_j = s_{i_j}$},\\
    \text{do not cross there} &\text{if $\sigma_j = e$}.
    \end{cases}
\end{equation*}
By this bijection, index $j$ is a defect of $\sigma$ in the sense of (\ref{eq:deodhardefect}) if and only if the paths meeting at $\ctr j$
have previously crossed an odd number of times.

Clearwater--Skandera 
extended this 
result~\cite[Cor.\,5.3]{CSkanTNNChar} to products of the form
\begin{equation}\label{eq:CSprod}
\wtc{s_{[a_1,b_1]}}q \cdots \wtc{s_{[a_m,b_m]}}q
= \sum_{v \in \sn} A_v(q) T_v,
\quad \text{where} \quad
    \wtc{s_{[a_j,b_j]}}q = \nTksp \sum_{u \leq s_{[a_j,b_j]}} \nTksp T_u,
\end{equation}
since reversals \avoidp.
This extension requires a more general definition of defects.
While the intersection of two paths in 
\eqref{eq:wirediagprod} is a union of vertices,
the intersection of two paths in
\begin{equation}\label{eq:bulletconcat}
F = F_{[a_1,b_1]} \bullet \cdots \bullet F_{[a_m,b_m]}
\end{equation}
is a subgraph of $F$ whose connected components are vertices or paths of the form 
\begin{equation}\label{eq:component}
(x_k,\dotsc, x_\ell)
\end{equation}
for some $k < \ell$.  
For each initial vertex $x_k$ in a component \eqref{eq:component} of the intersection of two paths, we will say that the paths {\em meet} at $x_k$.
Our embedding of star networks
in the plane naturally allows us to declare an edge entering (exiting) a vertex $x_k$ to be above or below another edge entering (exiting) $x_k$.
We will call a component \eqref{eq:component} in the intersection of paths $\pi_i$, $\pi_j$,
a {\em crossing} of $\pi_i$ and $\pi_j$ if the two paths enter $x_k$ and exit $x_\ell$ in different orders.
Extending the Billey--Warrington definition of defect to accommodate three or more paths passing through a vertex, we have the following~\cite[\S 5]{CSkanTNNChar}.

\begin{defn}\label{d: defect}
    Given a path family $\pi$ covering $F = F_{[a_1,b_1]} \bullet \cdots \bullet F_{[a_m,b_m]}$,
define a {\em defect} of $\pi$ at 
$\ctr k$ to be
a triple $(\pi_i, \pi_j, k)$ with
$i < j$
and $\pi_i$ and $\pi_j$ meeting at $\ctr k$ after having crossed an odd number of times. 
Define $\dfct(\pi)$ to be the number of defects of $\pi$.

\end{defn}

Extending the interpretation \eqref{eq:PiwF} of a planar network, we define the set
\begin{equation}\label{eq:Pi with v and d}
    \Pi_{v,d}(F) = \{ \pi \in \Pi_v(F) \,|\, \dfct(\pi) = d \}
\end{equation}
and we
say that $F$ {\em graphically
represents the element}
\begin{equation}\label{eq:Gtozhnq}
    \sum_{v \in \sn} \sum_{d \geq 0} |\Pi_{v,d}(F)|q^d T_v 
    = \nTksp\sum_{\pi \in \Pi(F)}\nTksp q^{\dfct(\pi)}T_{\type(\pi)}
\end{equation}
{\em as an element of $\hnq$}.
For $F = F_{[a_1,b_1]} \circ \cdots \circ F_{[a_m,b_m]}$, this element is
$\wtc{s_{[a_1,b_1]}}q \cdots \wtc{s_{[a_m,b_m]}}q$;
for $F = F_{[a_1,b_1]} \bullet \cdots \bullet F_{[a_m,b_m]}$, it is
$\wtc{s_{[a_1,b_1]}}q \cdots \wtc{s_{[a_m,b_m]}}q$, divided by
the product, over all edges $(x_i,x_j)$, of the
$q$-factorial polynomials $m(x_i,x_j)_q!$.  (See, e.g., \cite{StanEC1}.)

The denominator above can also be expressed in terms of intervals appearing in reversals.
Given a sequence of $m$ intervals
 \begin{equation}\label{eq:intervalseq}
 \mathcal A = ([a_1,b_1], \dotsc, [a_m,b_m]),
 \end{equation}
 define $\tbinom m2$ more intervals $\{I_{i,j} \,|\, i < j \}$ by
 \begin{equation}\label{eq:intervaloverlap}
 I_{i,j} = [a_i,b_i] \cap [a_j,b_j] \ssm ( [a_{i+1},b_{i+1}] \cup \cdots \cup [a_{j-1},b_{j-1}]).
 \end{equation}
Let $f_{\ntnsp\mathcal A}(q)$ be the product of the $q$-factorials of the cardinalities of the intervals (\ref{eq:intervaloverlap}),
 \begin{equation}\label{eq:overlappoly}
 f_{\ntnsp\mathcal A}(q) \defeq \prod_{i < j} |I_{i,j}|_q!.
 \end{equation}
Say that a Kazhdan--Lusztig basis element $\wtc wq$ has a {\em parabolic factorization} 
if there is a sequence (\ref{eq:intervalseq}) of intervals satisfying
\begin{equation}\label{eq:reversalfactor}
\wtc wq =
\frac1{f_{\ntnsp \mathcal A}(q)} \wtc{s_{[a_1,b_1]}}q \cdots \wtc{s_{[a_m,b_m]}}q.
\end{equation}
Two results on parabolic factorization are the following~\cite[Thm.\,1]{BWHex}, \cite[Thm.\,4.3]{SkanNNDCB}.
\begin{thm}\label{t:BWHex}
If $w \in \sn$ avoids the patterns $321$, $56781234$, $46781235$, $56718234$, $46718235$,
and if $s_{i_1} \cdots s_{i_\ell}$ is a reduced expression for $w$, then we have
$\wtc wq = \wtc {s_{i_1}}q \cdots \wtc {s_{i_\ell}}q$.
\end{thm}
\begin{thm}\label{t:SkanSmooth}
    If $w \in \sn$ \avoidsp, then there exists a sequence \eqref{eq:intervalseq} of intervals such that
    we have the factorization \eqref{eq:reversalfactor}.
\end{thm}
The combination of Theorems~\ref{t:BWHex}, \ref{t:SkanSmooth} is not the strongest result possible.
Indeed, the known parabolic factorization $\wtc{4231}q = \wtc{s_{[1,2]}}q \wtc{s_{[2,4]}}q \wtc{s_{[1,2]}}q$ is guaranteed by neither theorem.

\section{Reduction of defects}\label{s:defectremoval} 


Our defect reduction theorem
asserts that if a star network can be covered by a path family having $d$ defects, then it can also be covered by a path family of the same type having $d-1$ defects. In certain simple cases, we can easily produce a $(d-1)$-defect family from a $d$-defect family by swapping a pair of subpaths.
For example, consider the 
star network and path families
\begin{equation*}
F_{[1,2]} \circ F_{[1,2]} = 
\begin{tikzpicture}[scale=.7,baseline=25]
  \pgfmathsetmacro{\k}{2}  
  \pgfmathsetmacro{\n}{2}  
  \pgfmathsetmacro{\indexOffset}{0.4}
  \pgfmathsetmacro{\xstart}{0}
  \pgfmathsetmacro{\xstep}{0.5}

  \foreach \index in {1, ..., \n} {
    \node at (\xstart - \indexOffset, \index) {$\scriptstyle \index$};
    \node at (\xstart + \xstep*2*\k + \indexOffset, \index) {$\scriptstyle \index$};
  }
  
   \draw[-,  thick]
    (0,1) -- (0.5,1.5) -- (1,1) -- (1.5,1.5) -- (2,1);
    \draw[-, thick]
    (0,2) -- (0.5,1.5) -- (1,2) -- (1.5,1.5) -- (2,2);
    \fill (.5, 1.5) circle (1mm);  \fill (1.5, 1.5) circle (1mm);
    \fill (0, 1) circle (1mm);  \fill (0, 2) circle (1mm);
    \fill (2, 1) circle (1mm);  \fill (2, 2) circle (1mm);
    \node at (0.55, 1) {$\scriptstyle x_1$};  \node at (1.55, 1) {$\scriptstyle x_2$};
\end{tikzpicture}, 
\qquad
\pi = 
\begin{tikzpicture}[scale=.7,baseline=25]
  \pgfmathsetmacro{\k}{2}  
  \pgfmathsetmacro{\n}{2}  
  \pgfmathsetmacro{\indexOffset}{0.4}
  \pgfmathsetmacro{\xstart}{0}
  \pgfmathsetmacro{\xstep}{0.5}


  \foreach \index in {1, ..., \n} {
    \node at (\xstart - \indexOffset, \index) {$\scriptstyle \index$};
    \node at (\xstart + \xstep*2*\k + \indexOffset, \index) {$\scriptstyle \index$};
  }

   \draw[-, ultra thick, dashed]
    (0,1) -- (0.5,1.5) -- (1,2) -- (1.5,1.5) -- (2,2);
    \draw[-, very thick, color=blue]    
    (0,2) -- (0.5,1.5) -- (1,1) -- (1.5,1.5) -- (2,1);
     \fill (.5, 1.5) circle (1mm);  \fill (1.5, 1.5) circle (1mm);
    \fill (0, 1) circle (1mm);  \fill (0, 2) circle (1mm);
    \fill (2, 1) circle (1mm);  \fill (2, 2) circle (1mm);
     \node at (0.55, 1) {$\scriptstyle x_1$};  \node at (1.55, 1) {$\scriptstyle x_2$};
\end{tikzpicture}, 
\qquad \sigma = 
\begin{tikzpicture}[scale=.7,baseline=25]
  \pgfmathsetmacro{\k}{2}  
  \pgfmathsetmacro{\n}{2}  
  \pgfmathsetmacro{\indexOffset}{0.4}
  \pgfmathsetmacro{\xstart}{0}
  \pgfmathsetmacro{\xstep}{0.5}

  \foreach \index in {1, ..., \n} {
    \node at (\xstart - \indexOffset, \index) {$\scriptstyle \index$};
    \node at (\xstart + \xstep*2*\k + \indexOffset, \index) {$\scriptstyle \index$};
  }

   \draw[-, ultra thick, dashed]
    (0,1) -- (0.5,1.5) -- (1,1) -- (1.5,1.5) -- (2,2);
    \draw[-, very thick, color=blue]    
    (0,2) -- (0.5,1.5) -- (1,2) -- (1.5,1.5) -- (2,1);
     \fill (.5, 1.5) circle (1mm);  \fill (1.5, 1.5) circle (1mm);
    \fill (0, 1) circle (1mm);  \fill (0, 2) circle (1mm);
    \fill (2, 1) circle (1mm);  \fill (2, 2) circle (1mm);
     \node at (0.55, 1) {$\scriptstyle x_1$};  \node at (1.55, 1) {$\scriptstyle x_2$};
\end{tikzpicture},
\end{equation*}
with $\dfct(\pi) = 1$, 
$\sigma$ constructed from $\pi$ by swapping the two $x_1$-to-$x_2$ subpaths of $\pi$,
and $\dfct(\sigma) = 0$.
On the other hand,
this simple procedure does not always reduce defects by $1$.
Consider the star network and path families of type $3124$
\begin{equation*}
F_{[1,3]} \bullet F_{[3,4]} \bullet F_{[1,3]} \bullet F_{[3,4]} = 
\begin{tikzpicture}[scale=.5,baseline=30]
  \pgfmathsetmacro{\k}{4}  
  \pgfmathsetmacro{\n}{4}  
  \pgfmathsetmacro{\indexOffset}{0.4}
  \pgfmathsetmacro{\xstart}{0}
  \pgfmathsetmacro{\xstep}{0.5}

  \foreach \index in {1, ..., \n} {
    \node at (\xstart - \indexOffset, \index) {$\scriptstyle \index$};
    \node at (\xstart + \xstep*2*\k + \indexOffset, \index) {$\scriptstyle \index$};
  }
  \node at (1.5, 2.4) {$\scriptstyle _{(2)}$};

   \draw[-, thick]
    (0,1) -- (0.5,2) -- (1,3) -- (2,4) -- (3,4) -- (3.5,3.5)  -- (4,3);

    \draw[-, thick]
    (0,2) -- (0.5,2) -- (2.5,2) -- (3,1) -- (4,1);

    \draw[-, thick]
    (0,3) -- (0.5,2) -- (4,2);

    \draw[-, thick]
    (0,4) -- (1,4) -- (2.5,2) -- (4,4);






\fill (0, 1) circle (1.3mm);  \fill (0, 2) circle (1.3mm); \fill (0, 3) circle (1.3mm);  \fill (0, 4) circle (1.3mm);
\fill (4, 1) circle (1.3mm);  \fill (4, 2) circle (1.3mm); \fill (4, 3) circle (1.3mm);  \fill (4, 4) circle (1.3mm);
\fill (0.5, 2) circle (1.3mm);  \fill (1.43, 3.45) circle (1.3mm); \fill (2.5, 2) circle (1.3mm);  \fill (3.57, 3.45) circle (1.3mm);
   
\node at (0.75, 1.5) {$\scriptstyle x_1$};  
\node at (1.5, 4.15) {$\scriptstyle x_2$}; 
\node at (2.25, 1.5) {$\scriptstyle x_3$};
\node at (3.45, 4.1) {$\scriptstyle x_4$};
\end{tikzpicture}, 
\quad\pi = 
\begin{tikzpicture}[scale=.5,baseline=30]
  \pgfmathsetmacro{\k}{4}  
  \pgfmathsetmacro{\n}{4}  
  \pgfmathsetmacro{\indexOffset}{0.4}
  \pgfmathsetmacro{\xstart}{0}
  \pgfmathsetmacro{\xstep}{0.5}

  \foreach \index in {1, ..., \n} {
    \node at (\xstart - \indexOffset, \index) {$\scriptstyle \index$};
    \node at (\xstart + \xstep*2*\k + \indexOffset, \index) {$\scriptstyle \index$};
  }
   \node at (1.5, 2.4) {$\scriptstyle{_{(2)}}$};

   \draw[-, ultra thick, dashed]
    (0,1) -- (0.5,2.1) -- (1,3) -- (2,4) -- (3,4) -- (3.5,3.5) -- (4,3);

    \draw[-, ultra thick, color=green]
    (0,2) -- (0.5,1.9) -- (2.5,1.9) -- (3,1) -- (4,1);

    \draw[-, thick, dotted]
    (0,3) -- (0.5,2) -- (4,2);

    \draw[-, ultra thick, color=blue]
    (0,4) -- (1,4) -- (2.5,2.1) -- (4,4);

\fill (0, 1) circle (1.3mm);  \fill (0, 2) circle (1.3mm); \fill (0, 3) circle (1.3mm);  \fill (0, 4) circle (1.3mm);
\fill (4, 1) circle (1.3mm);  \fill (4, 2) circle (1.3mm); \fill (4, 3) circle (1.3mm);  \fill (4, 4) circle (1.3mm);
\fill (0.5, 2) circle (1.3mm);  \fill (1.43, 3.45) circle (1.3mm); \fill (2.5, 2) circle (1.3mm);  \fill (3.57, 3.45) circle (1.3mm);
   
\node at (0.75, 1.5) {$\scriptstyle x_1$};  
\node at (1.5, 4.15) {$\scriptstyle x_2$}; 
\node at (2.25, 1.5) {$\scriptstyle x_3$};
\node at (3.45, 4.1) {$\scriptstyle x_4$};
   
\end{tikzpicture}, \quad
\sigma = 
\begin{tikzpicture}[scale=.5,baseline=30]
  \pgfmathsetmacro{\k}{4}  
  \pgfmathsetmacro{\n}{4}  
  \pgfmathsetmacro{\indexOffset}{0.4}
  \pgfmathsetmacro{\xstart}{0}
  \pgfmathsetmacro{\xstep}{0.5}

  \foreach \index in {1, ..., \n} {
    \node at (\xstart - \indexOffset, \index) {$\scriptstyle \index$};
    \node at (\xstart + \xstep*2*\k + \indexOffset, \index) {$\scriptstyle \index$};
  }
   \node at (1.5, 2.4) {$\scriptstyle _{(2)}$};

   \draw[-, ultra thick, dashed]
    (0,1) -- (0.5,2.1) -- (1,3) -- (1.5, 3.5) -- (2.5, 2.1) -- (3.5, 3.5) -- (4,3);

    \draw[-, ultra thick, color=green]
    (0,2) -- (0.5,1.9) -- (2.5,1.9) -- (3,1) -- (4,1);

    \draw[-, thick, dotted]
    (0,3) -- (0.5,2) -- (4,2);

    \draw[-, ultra thick, color=blue]
    (0,4) -- (1,4) -- (1.5, 3.5) -- (2, 4) -- (3, 4) -- (3.5, 3.5) -- (4,4);

\fill (0, 1) circle (1.3mm);  \fill (0, 2) circle (1.3mm); \fill (0, 3) circle (1.3mm);  \fill (0, 4) circle (1.3mm);
\fill (4, 1) circle (1.3mm);  \fill (4, 2) circle (1.3mm); \fill (4, 3) circle (1.3mm);  \fill (4, 4) circle (1.3mm);
\fill (0.5, 2) circle (1.3mm);  \fill (1.43, 3.45) circle (1.3mm); \fill (2.5, 2) circle (1.3mm);  \fill (3.57, 3.45) circle (1.3mm);
   
\node at (0.75, 1.5) {$\scriptstyle x_1$};  
\node at (1.5, 4.15) {$\scriptstyle x_2$}; 
\node at (2.25, 1.5) {$\scriptstyle x_3$};
\node at (3.45, 4.1) {$\scriptstyle x_4$};
   
\end{tikzpicture},
\end{equation*}
with $\dfct(\pi) = 1$, and $\sigma$ constructed from $\pi$ by
swapping the $x_2$-to-$x_4$ subpaths of
$\pi_1$ and $\pi_4$.  The swap
eliminates the defect $(\pi_1, \pi_4, 4)$,
but introduces two more:
$(\sigma_1, \sigma_2, 3)$, $(\sigma_1, \sigma_3, 3)$.
Thus we have $\dfct(\sigma) = 2$.
There is in fact a type-$3124$ path family 
\begin{equation*}
\tau = 
\begin{tikzpicture}[scale=.5,baseline=30]
  \pgfmathsetmacro{\k}{4}  
  \pgfmathsetmacro{\n}{4}  
  \pgfmathsetmacro{\indexOffset}{0.4}
  \pgfmathsetmacro{\xstart}{0}
  \pgfmathsetmacro{\xstep}{0.5}

  \foreach \index in {1, ..., \n} {
    \node at (\xstart - \indexOffset, \index) {$\scriptstyle \index$};
    \node at (\xstart + \xstep*2*\k + \indexOffset, \index) {$\scriptstyle \index$};
  }
   \node at (1.5, 2.4) {$\scriptstyle _{(2)}$};

   \draw[-, ultra thick, dashed]
    (0,1) -- (0.5, 1.9) -- (2.5, 1.9)  -- (3.5, 3.5) -- (4,3);

    \draw[-, ultra thick, color=green]
    (0,2) -- (2.5,2) -- (3,1) -- (4,1);

    \draw[-, thick, dotted]
    (0,3) -- (0.5, 2.1) -- (1, 3) -- (1.5, 3.5) -- (2.5, 2)  -- (4,2);

    \draw[-, ultra thick, color=blue]
    (0,4) -- (1,4) -- (1.5, 3.5) -- (2, 4) -- (3, 4) -- (3.5, 3.5) -- (4,4);

\fill (0, 1) circle (1.3mm);  \fill (0, 2) circle (1.3mm); \fill (0, 3) circle (1.3mm);  \fill (0, 4) circle (1.3mm);
\fill (4, 1) circle (1.3mm);  \fill (4, 2) circle (1.3mm); \fill (4, 3) circle (1.3mm);  \fill (4, 4) circle (1.3mm);
\fill (0.5, 2) circle (1.3mm);  \fill (1.43, 3.45) circle (1.3mm); \fill (2.5, 2) circle (1.3mm);  \fill (3.57, 3.45) circle (1.3mm);
   
\node at (0.75, 1.5) {$\scriptstyle x_1$};  
\node at (1.5, 4.15) {$\scriptstyle x_2$}; 
\node at (2.25, 1.5) {$\scriptstyle x_3$};
\node at (3.45, 4.1) {$\scriptstyle x_4$};
   
\end{tikzpicture}
\end{equation*}
satisfying 
$\dfct(\tau) = \dfct(\pi)-1 = 0$, but the naive approach above does not produce it from $\pi$.

To describe the process of reducing defects in a path family by exactly $1$, we begin
by stating a map which modifies a path family by removing a defect at a specified vertex,
possibly creating earlier defects.
For $F = F_{[a_1,b_1]} \bullet \cdots \bullet F_{[a_m,b_m]}$ and $k\in\{2, \dotsc, m\}$, 
define
\begin{equation}\label{equ: phik}
\begin{aligned}
    \phi_k : \{\pi \in \Pi(F) \,|\, \pi \text{ has a defect at } x_k
    \} &\rightarrow \Pi(F)\\
    \pi &\mapsto \hat \pi\end{aligned}
    \end{equation}
by the following algorithm.
\begin{alg}\label{a:removeEarlyDefectRandomk}
    Given star network 
    $F = F_{[a_1, b_1]} \bullet 
    \cdots \bullet F_{[a_m, b_m]} \in \cnet n$, 
    path family $\pi \in \Pi(F)$,
    and index $k$ such that $\pi$ has a defect at $\ctr k$,
    do 
    \begin{enumerate}
    \item Let ($r, t$) be the lexicographically least pair such that $(\pi_r, \pi_t, k)$ is a defect.
    \item 
    Let $s$ be the largest index such that $\pi_s$ enters vertex $\ctr k$
    through the same edge as $\pi_r$ and $(\pi_s, \pi_t, k)$ is a defect.
    \item Let $\ctr l$ be the final vertex in the 
    unique crossing of 
    $\pi_s$ and $\pi_t$ prior to $\ctr k$. 
    \item Create a new path family $
    \hat \pi = (\hat \pi_1, \dotsc, \hat \pi_n)$ by \begin{enumerate}
        \item $\hat \pi_i = \pi_i$ if $i \not \in \{s,t\}$.
        \item $\hat \pi_s$ is $\pi_s$ with the
        $\ctr l$-to-$\ctr k$ subpath replaced by that of $\pi_t$.
        \item $\hat \pi_t$ is $\pi_t$ with the
        $\ctr l$-to-$\ctr k$ subpath replaced by that of $\pi_s$.
       \end{enumerate}
    \end{enumerate}
\end{alg}


\begin{prop}\label{p: one less defect condensed-adjusted}
Algorithm~\ref{a:removeEarlyDefectRandomk} produces a path family
$\phi_k(\pi)=\hat\pi$ that satisfies 
    \begin{enumerate}
        \item $\type(\hat \pi) = \type(\pi)$,
        \item for each $p > k$, we have $\{(i,j) \,|\, (\hat\pi_i, \hat\pi_j, p) \text{ defective}\} = \{(i,j) \,|\, (\pi_i, \pi_j, p) \text{ defective}\}$,      
        \item 
        $\#\{(i,j) \,|\, (\hat\pi_i, \hat\pi_j, k) \text{ defective}\} = \#\{(i,j) \,|\, (\pi_i, \pi_j, k) \text{ defective}\} - 1$. 
    \end{enumerate}
\end{prop}
\begin{proof}
Omitted.
 \end{proof}
We may further modify the path family $\hat \pi$ 
to produce a path family having no
defects at $x_1, \dotsc, x_{k-1}$.
\begin{prop}\label{p: zero defects condensed}
    Fix $F = F_{[a_1, b_1]} \bullet 
    \cdots \bullet F_{[a_m, b_m]} \in \cnet n$,
    and path family
    $\pi \in \Pi_v(F)$ having earliest 
    defect at $x_k$.
    Then there exists a path family $\sigma \in \Pi_v(F)$ which satisfies 
    \begin{enumerate}
        \item for all $p \neq k$, we have $\{(i,j) \,|\, (\sigma_i, \sigma_j, p) \text{ defective}\} = \{(i,j) \,|\, (\pi_i, \pi_j, p) \text{ defective}\}$,   
        \item for $p = k$, we have $\#\{(i,j) \,|\, (\sigma_i, \sigma_j, p) \text{ defective}\} = \# \{(i,j) \,|\, (\pi_i, \pi_j, p) \text{ defective}\}-1$.
    \end{enumerate}
\end{prop}
\begin{proof}
Omitted.
\end{proof}

By Proposition~\ref{p: zero defects condensed}, we see that a path family having type $v$ and $d$ defects (i.e., in $\Pi_{v,d}(F)$) implies the existence of other path families of type $v$ having $d-1,\dotsc, 0$ defects.

\begin{thm}\label{t:decrease defects by 1}
    Fix star network $F \in \cnet n$.  If for some $v \in \sn$ and $d \geq 1$ the set $\Pi_{v,d}(F)$ is nonempty, then the sets $\Pi_{v,d-1}(F),\dotsc, \Pi_{v,0}(F)$ are also nonempty, and $|\Pi_{v,0}(F)| = 1$.
\end{thm}
\begin{proof}
Omitted.
\end{proof}

\section{Main Result}\label{s:nofactor}

Every polynomial in $\mathbb N[q]$ with constant term $1$ arises as a Kazhdan--Lusztig polynomial~\cite{PoloConstrArb}.  Gaetz--Gao~\cite{GGMinPower} studied the sequences of coefficients in these polynomials, especially coefficients equal to 
$0$ 
between other
nonzero coefficients.
Define a function $\singdeg: \sn \rightarrow \mathbb N \cup \{\infty\}$ by
\begin{equation}\label{eq:gammadef}
\singdeg(w) = \begin{cases}
    \infty &\text{if $P_{e,w}(q) = 1$},\\
    \min\{ k > 0 \,|\, 
    \text{coefficient of $q^k$ in $P_{e,w}(q)$ is nonzero}\} 
    &\text{if $P_{e,w}(q) \neq 1$}.
    \end{cases}    
\end{equation}
This is 
a lower bound on degrees for which
Poincar\'e duality fails in the Schubert variety $X_w$, and can be computed in terms of patterns in $w$ and a related definition.
Specifically, given $w \in \sn$
not avoiding the pattern $3412$,
define the {\em 
3412-gap} of $w$
by
\begin{equation}\label{eq:cortezgap}
\pgap(w) = \min\{ w_{i_1} - w_{i_4} \,|\, \text{subword $w_{i_1} w_{i_2} w_{i_3} w_{i_4}$ 
matches the pattern $3412$} \}.
\end{equation}
For $w$ \avoidingp, we have $\singdeg(w) = \infty$.  Otherwise, we can compute $\singdeg(w)$ in terms of $\pgap(w)$ as follows
~\cite[Thm.\,1.6]{GGMinPower}.
\begin{thm}\label{t:GaetzGao}
    For $w$ not \avoidingp{}
    we have
    \begin{equation*}
        \singdeg(w) = \begin{cases}
            \pgap(w) &\text{if $w$ avoids the pattern $4231$},\\
            1 &\text{otherwise}.
        \end{cases}
    \end{equation*}
\end{thm}


For example, consider the permutation $45312 \in \mfs 5$, 
which 
avoids the pattern $4231$ and has
only the subword $4512$ matching the pattern $3412$. Since $\pgap(45312)=2$, Theorem~\ref{t:GaetzGao} implies that the coefficient of $q$ in $P_{e,45312}(q)$ is $0$ and that the coefficient of $q^2$ is not.
This is consistent with the fact that $P_{e,45312}(q) = 1+q^2$. (See \cite[p.\,75]{BilleyLak}.)

For each permutation $w$ having $\singdeg(w) > 1$,
the Kazhdan--Lusztig basis element $\smash{\wtc wq}$ has no parabolic factorization 
\eqref{eq:reversalfactor} and therefore is not graphically representable 
by a star network, in the sense of (\ref{eq:Gtozhnq}).
\begin{thm}\label{t:main}
    For $w \in \sn$ 
    avoiding the pattern $4231$, 
    not avoiding the pattern $3412$,
    and having $\pgap(w) > 1$, 
    the Kazhdan--Lusztig basis element $\wtc wq$ 
    has no parabolic factorization.
\end{thm}
\begin{proof}
Fix $w$ as above with $k = \pgap(w)$ and suppose that the star network $F$
graphically represents $\smash{\wtc wq}$ as an element of $\hnq$,
\begin{equation}\label{eq:main}
\wtc wq = \sum_{v \in \sn} \sum_{d \geq 0} |\Pi_{v,d}(F)|
q^d T_v = \sum_{v \leq w} P_{v,w}(q) T_v.
\end{equation}
Since the constant term of $P_{e,w}(q)$ is $1$,
we have $|\Pi_{e,0}(F)| = 1$.
By our definition of $k$
and Theorem~\ref{t:GaetzGao}, the coefficients of $q, \dotsc, q^{k-1}$
in $P_{e,w}(q)$ are $0$, while the coefficient of $q^k$ is positive.
In particular, we have the cardinalities $|\Pi_{e,1}(F)| = \cdots = |\Pi_{e,k-1}(F)| = 0$ and $|\Pi_{e,k}(F)| > 0$, which
contradict
Theorem~\ref{t:decrease defects by 1}.
\end{proof}
By Theorem~\ref{t:decrease defects by 1}, parabolic factorization of $\wtc wq$ implies that {\em none} of the Kazhdan--Lusztig polynomials $P_{v,w}(q)$ has internal coefficients equal to zero.
\begin{thm}
    If $\wtc wq$ has a parabolic factorization,
    then for every $v < w$ there exists $k = k(v)$
    in $\mathbb N$ 
    such that we have $P_{v,w}(q) = 1 + a_1q + \cdots + a_kq^k$ with $a_1,\dotsc, a_k > 0$.
\end{thm}
It is easy to show that the inequality $\pgap(w) > 1$ implies that $w$ does not avoid the pattern $45312$. It is also easy to show that no star network graphically represents
$\wtc{453129786}q$, even though the subword $9786$ matches the pattern $4231$. Indeed, some limited experimentation~\cite{DatSkan} suggests that avoidance of 
the pattern $45312$ is important and avoidance of the pattern $4231$ is unimportant in the classification of permutations $w$ for which
$\smash{\wtc wq}$ is graphically representable by a star network.
We conjecture the following partial answer to \cite[Quest.\,4.5]{SkanNNDCB}.
\begin{conj}
    If $w \in \sn$ does not avoid the pattern $45312$, then the Kazhdan--Lusztig basis element $\wtc wq$ has no parabolic factorization.
\end{conj}

\bibliographystyle{eptcs}
\bibliography{my}
\end{document}